\documentclass[runningheads]{llncs}

\usepackage{nicematrix}
\usepackage{dsfont}
\usepackage{tikz}
\usepackage{xcolor}
\usepackage[T1]{fontenc}
\usepackage{subcaption}

\colorlet{lblue}{blue!50!white}
\colorlet{lred}{red!50!white}
\colorlet{lgreen}{green!50!white}
\colorlet{lpurple}{purple!50!white}
\colorlet{lorange}{orange!50!white}
\colorlet{lpink}{pink!50!white}
\colorlet{lbrown}{brown!50!white}
\colorlet{lyellow}{yellow!50!white}
\colorlet{lolive}{olive!50!white}

\usepackage{amssymb}
\usepackage{amsmath}

\usepackage[hidelinks]{hyperref}

\usepackage{graphicx}

\title{Decomposing Words for Enhanced Compression:\\ Exploring the Number of Runs in the Extended Burrows-Wheeler Transform}

\titlerunning{Decomposing Words for Enhanced Compression}

\author{Florian Ingels\orcidID{0000-0002-8556-0087} \and
Anaïs Denis \and
Bastien Cazaux\orcidID{0000-0002-1761-4354}}
\authorrunning{F. Ingels, A. Denis and B. Cazaux}
\institute{Univ. Lille, CNRS, Centrale Lille, UMR 9189 CRIStAL, F-59000 Lille, France
\email{\{florian.ingels,bastien.cazaux\}@univ-lille.fr}}

\newcommand{\size}[1]{\vert #1 \vert}

\begin{document}

\maketitle          

\begin{abstract}

The Burrows-Wheeler Transform (BWT) is a fundamental component in many data structures for text indexing and compression, widely used in areas such as bioinformatics and information retrieval. The extended BWT (eBWT) generalizes the classical BWT to multisets of strings, providing a flexible framework that captures many BWT-like constructions. Several known variants of the BWT can be viewed as instances of the eBWT applied to specific decompositions of a word. A central property of the BWT, essential for its compressibility, is the number of maximal ranges of equal letters, named runs. In this article, we explore how different decompositions of a word impact the number of runs in the resulting eBWT. First, we show that the number of decompositions of a word is exponential, even under minimal constraints on the size of the subsets in the decomposition. Second, we present an infinite family of words for which the ratio of the number of runs between the worst and best decompositions is unbounded, under the same minimal constraints. These results illustrate the potential cost of decomposition choices in eBWT-based compression and underline the challenges in optimizing run-length encoding in generalized BWT frameworks.

\keywords{Burrows-Wheeler Transform \and Extended Burrows-Wheeler Transform \and Word Decompositions \and Combinatorics on Words}
\end{abstract}

\section{Introduction}

Text compression issues are ubiquitous in many fields, especially bioinformatics, where the volume of strings to be stored is growing exponentially fast \cite{katz2022sequence}. A widespread method, \emph{run-length encoding} (RLE) \cite{robinson1967results}, consists of replacing consecutive ranges of identical characters by pairs (character, range length). Those ranges are named \emph{runs}, and their number, denoted by $\texttt{runs}(w)$, is a good measure of the compressibility of a string $w$. However, this technique is only really interesting if the data has a low number of ranges. To this end, RLE approaches are almost always coupled with the Burrows-Wheeler transform (BWT) \cite{burrows1994block}, a bijective transform of strings, which has the interesting property of producing long runs when the starting string contains highly repetitive substrings. This property, called the \emph{clustering effect}, has been extensively studied by the literature, see for instance \cite{manzini2001analysis,mantaci2017measuring,biagi2025number,navarro2020indexing,giuliani2021novel,akagi2023sensitivity}. Therefore, the compressibility of a string $w$ is given by $\texttt{runs}(\texttt{bwt}(w))$.

The extended Burrows-Wheeler transform (eBWT) use a similar principle to the BWT, and transforms bijectively not just a single string, but a collection of strings \cite{mantaci2007extension}, into a unique transformed string that can afterwards be compressed using RLE. Note that the eBWT is indeed an extension of the BWT as they coincide when the collection is reduced to a single element, i.e. $\texttt{ebwt}(\lbrace w\rbrace) = \texttt{bwt}(w)$. Although introduced in 2007, eBWT has not really been as widely and deeply embraced as BWT, perhaps notably because it was not until 2021 that a linear-time eBWT construction algorithm was proposed \cite{bannai2021constructing}. As a result, very little work has been devoted to studying the properties of eBWT, particularly from a theoretical point of view -- see nevertheless  \cite{cenzato2024survey}.

This article focuses on a natural question that has never been studied in the literature, to the best of our knowledge: given a string $w$, is it possible to decompose it into a multiset of strings $w_1,\dots,w_m$, so that $w = w_1\cdots w_m$ and
$$\texttt{runs}(\texttt{ebwt}(\lbrace\!\lbrace w_1,\dots,w_m\rbrace\!\rbrace))\leq \texttt{runs}(\texttt{bwt}(w))\quad ?$$
In other words, is it possible to decompose a string so that the eBWT of its decomposition is more compressible than the BWT of the string itself? Can we find an optimal decomposition that minimises the number of runs?

It is worth noting that this idea of ``decompose to compress'' already exists in the literature, via the so-called bijective Burrows-Wheeler transform (BBWT) \cite{gil2012bijective,bannai2021constructing}. It is known that any string $w$ can be uniquely represented by its Lyndon factorization, i.e. there exists unique $w_1,\dots, w_l$ so that $S = w_1\cdots w_l$, with $w_1 \geq \cdots\geq w_l$ (for the standard lexicographical order) and all $w_i$'s are Lyndon words (i.e. $w_i$ is strictly smaller than any of its non-trivial cyclic shifts) \cite{chen1958free}. The idea behind the BBWT is very simple: use the eBWT of this Lyndon factorisation, whose uniqueness ensures the bijectivity of the final transformation; i.e. $\texttt{bbwt}(w) = \texttt{ebwt}(\texttt{lyndon}(w))$.

For the purposes of this article, the Lyndon factorization of a string is just one particular decomposition in the space of all possible decompositions. Moreover, to the best of our knowledge, there is no guarantee that BBWT is more compressible than BWT. In fact, there exist an infinite family of strings for which the BBWT is significantly more compressible than the BWT~\cite{BadkobehBK24}, and conversely, an infinite family of strings where the BWT yields a much better compression than the BBWT~\cite{BannaiIN24}. Another particular -- and trivial -- decomposition is the string itself (since the eBWT and BWT coincide on singletons). All other possible decompositions are, at this stage, \emph{terra incognita}.

The complexity of the problem we investigate is not clear, and we leave open the question of its possible NP-completeness. However, in this article we show:
\begin{enumerate}
    \item that there is an exponential number of possible decompositions, and therefore that brute force is doomed to failure, without great surprise;
    \item that the number of runs of the best possible decomposition of $w$ is bounded by a quantity that does not depend on $\size{w}$, but rather on the minimal size of the subsets of this decomposition -- which means that there is potentially a lot to be gained by searching for a good decomposition;
    \item that there is an infinite family of strings for which the ratio of the number of runs between the worst and best decompositions is not bounded, even under minimal constraints on the size of the subsets in the decomposition  -- in other words, there is also potentially a lot to lose if we decompose without strategy. We use the notion of ratio in a similar spirit to~\cite{GiulianiILPST21}, where the authors show that the ratio between the number of runs in the BWT of a string and that of its reverse is unbounded.
\end{enumerate}

These three points combined justify, for us, the interest in studying this problem, and we hope to raise the curiosity of the community as to its possible resolution or complexity. The remainder of this article is organized as follows:
\begin{itemize}
    \item Section~\ref{sec:definition} provides a precise overview of the concepts discussed in this introduction, including the problem we are investigating.
    \item Section~\ref{sec:main} details the main results of this article, whose proofs are  spread out in Sections~\ref{sec:cardinal},~\ref{sec:special_linear_partition}, and~\ref{sec:antecedents_ba}.
\end{itemize}

\section{Preliminaries}\label{sec:definition}

Given an alphabet $\Sigma = \lbrace a_1,\dots,a_{\sigma}\rbrace$, we use the Kleene operator, i.e. $\Sigma^{\ast}$, to define the set of finite strings over $\Sigma$. For a string $w = a_1 a_2  \cdots a_{n}$, we denote by $\size{w}$ the size of $w$, i.e $\size{w} = n$. The set of all strings of size $n$ is denoted by $\Sigma^n$. The lexicographical order $\prec_{\text{lex}}$ on $\Sigma^n$ is defined as follows: for any two strings $x=a_1\cdots a_n$ and $y=b_1\cdots b_n$, $x \prec_{\text{lex}} y$ if and only if there exists $1 \leq i \leq n$ so that $a_j=b_j$ for all $j< i$ and $a_i<b_i$. The \emph{number of runs} of a string $w=a_1\cdots a_n$, denoted by $\texttt{runs}(w)$, is defined as $\texttt{runs}(w)=\sum_{i=1}^{n-1} \mathds{1}_{a_i \neq a_{i+1}}.$

A string $w'$ is a \emph{circular rotation} of another string $w$ if and only if there exist two strings $u,v\in \Sigma^\ast$ so that $w'=u\cdot v$ and $w=v\cdot u$. For any string $w$, the Burrows-Wheeler transform (BWT) of $w$, denoted by $\texttt{bwt}(w)$, is obtained by concatenating the last characters of the $\size{w}$ circular rotations of $w$, sorted by ascending lexicographical order \cite{burrows1994block}. This transformation is bijective, as $w$ can be reconstructed from $\texttt{bwt}(w)$ --- up to a cyclic rotation.

A string $w$ is said to be periodic if there exist $v\in \Sigma^\ast$ and $n\geq 2$ so that $w=v^n$; otherwise $w$ is said to be primitive. For any string $w$, there exist a unique primitive string $v$, denoted $\text{root}(w)$ and a unique integer $k$, denoted $\exp(w)$, so that $w=v^k$ \cite{lothaire1997combinatorics}. We define the $\omega$-order as follows: for any two strings $u,v\in \Sigma^\ast$, we denote by $u^\omega = u\cdot u\cdot u\cdots$ and $v^\omega = v\cdot v\cdot v\cdots$ the infinite concatenations of $u$ and $v$; then, we say that $u\prec_{\omega}v$ if and only if either (i) $\exp(u)\leq \exp(v)$ if $\text{root}(u)=\text{root}(v)$, or (ii)  $u^\omega\prec_{\text{lex}} v^\omega$ otherwise. Note that $u\prec_{\omega} v \iff u\prec_{\text{lex}}v$ whenever $\size{u}=\size{v}$. Provided a multiset of strings $W=\lbrace\!\lbrace w_1,\dots,w_m\rbrace\!\rbrace$, the extended Burrows-Wheeler transform (eBWT) of $W$, denoted by $\texttt{ebwt}(W)$, is obtained by concatenating the last characters of the $\size{w_1}+\cdots+\size{w_m}$ circular rotations of $w_1,\dots,w_m$, sorted by ascending $\omega$-order \cite{mantaci2007extension}. When arranging these circular rotation into a matrix, the first and last column are usually denoted by $F$ and $L$ and corresponds, respectively, to the letters of $W$ arranged in increasing order, and to $\texttt{ebwt}(W)$. Note that if applied to a singleton, the eBWT coincide with the BWT, i.e. $\texttt{ebwt}(\lbrace w\rbrace) = \texttt{bwt}(w)$.

\begin{example}\label{exple:inverting_ebwt}
Let $W = \lbrace\!\lbrace baa, bab\rbrace\!\rbrace$. The cyclic rotations of the strings of $W$ are : $baa,aba,aab, bab, abb$ and $bba$. Arranging these strings in ascending $\omega$-order leads to the matrix of Figure~\ref{fig:ebwt_calcul}, where $F = aaabbb$ and $L = \texttt{ebwt}(W) = bababa$.
\end{example}

The eBWT is also bijective, as $W$ can also be reconstructed from $\texttt{ebwt}(W)$, up to a cyclic rotation of each string $w_1,\dots,w_m$. Remember that $\texttt{ebwt}(W)$ corresponds to the last column $L$ of the eBWT matrix, consisting of all circular rotations of the strings composing $W$, arranged in increasing $\omega$-order. The first column $F$ can be easily reconstructed, by sorting the characters of $L$ in increasing lexicographical order. We have the following facts -- see for instance \cite[Proposition~2.1]{mantaci2017measuring}.

\begin{proposition}\label{prop:inverting_ebwt}
\begin{enumerate}
    \item For any row $j$ in the eBWT matrix, the letter $F[j]$ cyclically follows $L[j]$ in some of the original strings $w_1,\dots, w_m$
    \item For each letter $a$, the $i$-th occurrence of $a$ in $F$ corresponds to the $i$-th occurrence of $a$ in $L$;
\end{enumerate}
\end{proposition}

We number each character of $F$ and $L$ by its occurrence rank among all identical characters (i.e. the first $a$ is denoted $a_1$, the second $a_2$, the first $b$ is $b_1$, the second $b_2$, and so on). Then, using Proposition~\ref{prop:inverting_ebwt}, we can invert the eBWT by identifying cycles of letters, as shown in Figure~\ref{fig:ebwt_invert}: starting from the first letter $L[1]$, we get the previous letter $F[1]$ (using item 1), then identify it back in $L$ (using item 2), get the previous letter, identify back, and so on, until we cycle back to $L[1]$. If there is any remaining letter not already part of the cycle, we start again the process with this letter, until all cycles are identified, corresponding to the strings $w_1,\dots,w_m$ -- up to a cyclic rotation.

\begin{figure}[h!]
    \centering
\begin{subfigure}[t]{0.4\textwidth}
\centering
    \begin{tikzpicture}[yscale=0.4]

\node at (0,6) {$F$};
\node at (2,6) {$L$};

\draw (-0.25,5.5)--(2.25,5.5);

\node (fa1) at (0,5) {$a$};
\node (fa2) at (0,4) {$a$};
\node (fa3) at (0,3) {$a$};
\node (fb1) at (0,2) {$b$};
\node (fb2) at (0,1) {$b$};
\node (fb3) at (0,0) {$b$};

\node (fa1) at (1,5) {$a$};
\node (fa2) at (1,4) {$b$};
\node (fa3) at (1,3) {$b$};
\node (fb1) at (1,2) {$a$};
\node (fb2) at (1,1) {$a$};
\node (fb3) at (1,0) {$b$};

\node (lb1) at (2,5) {$b$};
\node (la1) at (2,4) {$a$};
\node (lb2) at (2,3) {$b$};
\node (la2) at (2,2) {$a$};
\node (lb3) at (2,1) {$b$};
\node (la3) at (2,0) {$a$};

\end{tikzpicture}
        \caption{Computing the eBWT matrix}
    \label{fig:ebwt_calcul}
\end{subfigure}~
\begin{subfigure}[t]{0.5\textwidth}
\centering
    \begin{tikzpicture}[yscale=0.4]

\node at (0,6) {$F$};
\node at (2,6) {$L$};

\draw (-0.25,5.5)--(2.25,5.5);

\node (fa1) at (0,5) {$a_1$};
\node (fa2) at (0,4) {$a_2$};
\node (fa3) at (0,3) {$a_3$};
\node (fb1) at (0,2) {$b_1$};
\node (fb2) at (0,1) {$b_2$};
\node (fb3) at (0,0) {$b_3$};

\node (lb1) at (2,5) {$b_1$};
\node (la1) at (2,4) {$a_1$};
\node (lb2) at (2,3) {$b_2$};
\node (la2) at (2,2) {$a_2$};
\node (lb3) at (2,1) {$b_3$};
\node (la3) at (2,0) {$a_3$};

\tikzstyle{arc}=[->, thick,lred]
\tikzstyle{retour}=[->, thick,lred,dotted]

\draw[arc] (lb1)--(fa1);
\draw[retour] (fa1)--(la1);
\draw[arc] (la1)--(fa2);
\draw[retour] (fa2)--(la2);
\draw[arc]  (la2)--(fb1);
\draw[retour] (fb1)--(lb1);

\tikzstyle{arc}=[->,thick,lblue]
\tikzstyle{retour}=[->, thick,lblue,dotted]

\draw[arc] (lb2)--(fa3);
\draw[retour] (fa3)--(la3);
\draw[arc] (la3)--(fb3);
\draw[retour] (fb3)--(lb3);
\draw[arc] (lb3)--(fb2);
\draw[retour] (fb2)--(lb2);

\end{tikzpicture}
        \caption{Inverting the eBWT: we get two cycles, leading to the strings \colorbox{lred}{$b_1a_1a_2$} and \colorbox{lblue}{$b_2a_3b_3$}.}
    \label{fig:ebwt_invert}
\end{subfigure}
    \caption{Example of computation and inversion of the eBWT of $W=\lbrace\!\lbrace baa,bab\rbrace\!\rbrace$.}
    \label{fig:ebwt_example}
\end{figure}

A string decomposition $W = \lbrace\!\lbrace w_1,\ldots, w_m\rbrace\!\rbrace$ of a string $w$ is a multiset of strings (possibly with duplicates) where the concatenation of the strings of $W$ corresponds to $w$, i.e. $w= w_{1}\cdots w_{m}$. We denote by $\mathcal{D}(w)$ the set of all possible decompositions of $w$. In this article, we are especially interested in decompositions $W=\lbrace\!\lbrace w_1,\dots,w_m\rbrace\!\rbrace$ where $\forall i,\,\size{w_i}>k$ for some integer $k\geq 1$. In such a case, we call $W$ a $k$-restricted decomposition. The set of all $k$-restricted decompositions of a string $w$ is denoted by $\mathcal{D}_k(w)$ --- with $\mathcal{D}_0(w)=\mathcal{D}(w)$.

As mentioned in the introduction, we are interested in this article in how one can decompose a string $w$ into a multiset of strings $w_1,\dots,w_m$ so that $$\texttt{runs}(\texttt{ebwt}(\lbrace\!\lbrace w_1,\dots,w_m\rbrace\!\rbrace))\leq \texttt{runs}(\texttt{bwt}(w)).$$

As a shortcut, we denote $\texttt{runs}(\texttt{ebwt}(\cdot))$ by $\rho(\cdot)$, so that we can rewrite this equation as $\rho(\lbrace\!\lbrace w_1,\dots,w_m\rbrace\!\rbrace) \leq \rho(w).$

There is an obvious decomposition, which consists of decomposing $w$ into as many one-letter strings as $\size{w}$, so that the eBWT of the resulting set is simply the letters of $w$ sorted in lexicographical order, and so the number of runs equals the number of different letters in $w$, which is optimal. If one wants to reconstruct $w$ by inverting the eBWT of a decomposition $w_1,\dots,w_m$, one must be able to recover, on the one hand, the original circular rotations of the strings and, on the other hand, their original order. While these practical considerations are beyond the scope of this article, they highlight why the trivial decomposition proposed above is of no practical interest. As a way to constrain the problem and get rid of this case, we propose to consider $k$-restricted decompositions.

We now formally introduce our problem of interest:

\begin{problem}\label{problem}
Provided $k\geq 1$ and $w\in\Sigma^\ast$, find $W\in\mathcal{D}_k(w)$ so that $\rho(W) \leq \rho(w)$. Alternatively, find $W\in\mathcal{D}_k(w)$ such that $\rho(W)$ is minimal.
\end{problem}

\section{Main results}\label{sec:main}

As mentioned in the introduction, we do not intend to propose an algorithm (or a heuristic) to solve Problem~\ref{problem} in this article, in the same way that its possible NP-completeness is left open. However, we propose three results which, in our view, justify studying this problem in further research; we also hope that the community will find interest and engage with these questions.

First of all, and without much surprise, exploring all the possible decompositions is doomed to failure, as a result of combinatorial explosion.
\begin{theorem}\label{thm:cardinal}
For any $k\geq 1$, there exist a constant $r >1$ and a complex polynomial $P\in \mathbb{C}[X]$ so that $\size{\mathcal{D}_k(w)}\underset{n\to\infty}{\sim} |P(n)|\cdot r^n$, for any string $w\in \Sigma^n$.
\end{theorem}
\begin{proof}
The proof is deferred to Section~\ref{sec:cardinal}.
\end{proof}

Nevertheless, the next result shows that finding an optimal decomposition can lead to a number of runs that is independent of the size of the initial string, and therefore highlights the potential gain in terms of compressibility.
\begin{theorem}\label{thm:best}
For any $k\geq 1$ and any string $w$, we have $$\displaystyle\min_{W\in\mathcal{D}_k(w)} \rho(W) \leq \sigma^{k+1}+4k+2.$$
\end{theorem}
\begin{proof}
The proof is deferred to Section~\ref{sec:special_linear_partition}.
\end{proof}

Finally, to highlight the potential loss of decomposing without any particular strategy, we show in the next result that there is an infinite family of strings for which the ratio between the worst decomposition and the best is unbounded. 

\begin{theorem}\label{thm:ratio}
For any $M\geq 0$ and any $k\geq 1$, there exists $w\in \Sigma^\ast$ so that $$\frac{\displaystyle\max_{W\in \mathcal{D}_k(w)} \rho(W)}{\displaystyle\min_{W\in\mathcal{D}_k(w)}\rho(W)}\geq M.$$
\end{theorem}
\begin{proof}
Using Theorem~\ref{thm:best}, it actually suffices to find a string $w$ so that
$$\max_{W\in \mathcal{D}_k(w)} \rho(W)\geq M\cdot \left(\sigma^{k+1}+4k+2\right).$$

In upcoming Section~\ref{sec:antecedents_ba}, we show that, for any $k\geq 1$, there exist a infinite family of strings $w\in\Sigma^\ast$ for which there exists $W\in \mathcal{D}_k(w)$ so that $\rho(W)=\size{w}-1$, which is maximal. Therefore it suffices to choose any string $w$ from said family so that $\size{w}-1\geq M\cdot \left(\sigma^{k+1}+4k+2\right)$.
\end{proof}

It is worth noting that Theorem~\ref{thm:ratio} is proven in the case k = 0 by~\cite{BannaiIN24} and~\cite{BadkobehBK24} by comparing two specific decompositions: the trivial decomposition (BWT) and the Lyndon factorization (BBWT).

As a conclusion, we hope that the combination of these three results proves the relevance of studying Problem~\ref{problem}. In anticipation of further research, we offer interested readers an online tool for exploring the possible decompositions of a string: \url{http://bcazaux.polytech-lille.net/EBWT/}.

\section{On the number of $k$-restricted decompositions}\label{sec:cardinal}

The goal of this section is to prove Theorem~\ref{thm:cardinal}, that is, to quantify the cardinality of $\mathcal{D}_k(w)$ and to find an asymptotic equivalent of this cardinality.

Let $k\geq 1$ and $w\in \Sigma^n$ for some $n\geq k+1$. Let $W\in \mathcal{D}(w)$ be a decomposition of $w$, i.e. $W = \lbrace\!\lbrace w_1,\dots,w_p\rbrace\!\rbrace$ and $w = w_1\cdots w_p$. Denoting by $a_1,\dots, a_n$ the letters of $w$, and $t_i=\size{w_i}$, notice that $w_1 = a_1\cdots a_{t_1}$, $w_2= a_{t_1+1}\cdots a_{t_1+t_2}$, and more generally $$w_i = a_{1 + t_1 + \cdots + t_{i-1}}\cdots a_{t_1+\cdots+t_i}.$$

Since the letters $a_1,\dots, a_n$ are fixed, any decomposition $W\in\mathcal{D}(w)$ is therefore entirely described by the ordered list of number $t_1,\dots,t_p$, with $t_1+\cdots+t_p = n$. Such an ordered list is called a \emph{composition} of $n$. A restricted composition is a composition where additional constraints are added on the $t_i$'s; for instance $t_i \in A$ for some subset $A\subset \mathbb{N}$ \cite{heubach2004compositions}. In our context, we are interested in restricted compositions where $t_i\geq k+1$ -- that we call \emph{$(k+1)$-restricted compositions}. We denote by $C(n,k)$ the number of $k$-restricted compositions of $n$ and by $C(n,k,p)$ the number of $k$-restricted compositions of $n$ with exactly $p$ summands. It is clear that (i) $\size{\mathcal{D}_k(w)} = C(n,k+1)$ -- again with $\size{w}=n$ -- and (ii) $C(n,k) = \displaystyle\sum_{p=1}^{\lfloor \frac{n}{k}\rfloor} C(n,k,p)$. We easily have $C(n,k,p) = \binom{n-pk+p-1}{p-1}$, using a stars and bars arguments -- see also \cite{jaklivc2010closed}. Therefore,
$$C(n,k) = \sum_{p=1}^{\lfloor \frac{n}{k}\rfloor} \binom{n-kp+p-1}{p-1}\underset{j=p-1}{=} \sum_{j=0}^{\lfloor \frac{n-k}{k}\rfloor} \binom{n-k-kj+j}{j}.$$

Harris \& Styles proved in \cite{harris1964generalization} that $\displaystyle\sum_{p=0}^{\lfloor \frac{n}{c}\rfloor} \binom{n-pc+p}{p}=G_n^c$;
where $G_n^c$ designates the $n$-th generalized Fibonacci number \cite{bicknell1996classes}, defined as follows: for any integer $c\geq 1$, $G_0^c = \dots = G_{c-1}^c = 1$ and for $n\geq c$, $G_n^c = G^c_{n-1}+G^c_{n-c}$. 


Combining this result with (i), we get the following.
\begin{proposition}\label{prop:compositions_cardinal}
For $k\geq 1$, $n\geq k+1$, and $w\in\Sigma^n$, $\size{\mathcal{D}_k(w)} = G_{n-(k+1)}^{k+1}$.
\end{proposition}

Let $r_1,\dots,r_e$ be the (distinct) complex roots of $X^c - X^{c-1}-1$. Then, there exists complex polynomials $P_1,\dots,P_e$ and a sequence $z_n$, which is zero for $n\geq c$, so that
$$G_n^c = z_n + P_1(n)\cdot r_1^n + \dots + P_e(n)\cdot r_e^n \quad\cite{brousseau1971linear}.$$

Note that despite $P_1,\dots,P_e$ and $r_1,\dots,r_e$ being complex polynomials and roots, the above formula does indeed yield an integer. To provide an asymptotic behaviour for $G_n^c$, we need the following result.

\begin{lemma}
There exists a complex root $r$ of $X^c - X^{c-1}-1$ so that $|r|>1$.
\end{lemma}
\begin{proof}
The Mahler measure of a polynomial $P(X) = a \cdot (X-r_1)\cdots (X-r_c)$ is defined as $\mathcal{M}(P) = |a|\cdot \prod_{i=1}^c \max(1,|r_i|)$. To prove our result, it is sufficient to prove that $\mathcal{M}(X^c - X^{c-1} -1) > 1$ -- since $a=1$ in our case. Smyth \cite{smyth2007mahler} proved that if $P$ is not reciprocal (i.e. $P(X) \neq X^c P(1/X)$) then $\mathcal{M}(P) \geq M(X^3-X-1) \approx 1.3247$. Since $X^c-X^{c-1}-1$ is not reciprocal, the conclusion holds.\end{proof}

Without loss of generality, suppose $r_1$ is the complex root of $X^c - X^{c-1}-1$ of maximum modulus -- with $|r_1|>1$ by the previous lemma. Then, when $n\to\infty$, we have $G_n^c \sim |P_1(n)|\cdot |r_1|^n$. To finish the proof of Theorem~\ref{thm:cardinal}, we use Proposition~\ref{prop:compositions_cardinal} to obtain
$$\size{\mathcal{D}_k(w)} \sim |P(n-(k+1))|\cdot |r|^{n-(k+1)}$$
where $P$ and $r$ correspond to the aforementioned polynomial $P_1$ and root $r_1$ when $c=k+1$.


\section{On the best $k$-restricted decomposition}\label{sec:special_linear_partition}

The goal of this section is to prove Theorem~\ref{thm:best}, that is, for any string $w$, and any integer $k\geq 1$, $\displaystyle\min_{W\in\mathcal{D}_k(w)}\rho(W) \leq \sigma^{k+1}+4k+2.$

\subsection{An important property of the eBWT}

\begin{proposition}\label{prop:remove_one_occurrence}
Let $A$ be a multiset of strings, and let $w\in A$ be some string with multiplicity $m\geq 1$. Let $B$ be the multiset of strings obtained from $A$ by removing one occurrence of $w$. Then
\begin{enumerate}
    \item $\rho(A)=\rho(B)$ if $m\geq 2$,
    \item $0\leq \rho(A) - \rho(B)\leq 2\cdot \size{w}$ otherwise.
\end{enumerate}
\end{proposition}
\begin{proof}
(1) Suppose first that $m\geq 2$. Therefore, after removing one occurrence of $w$ from $A$ to obtain $B$, there remains at least one occurrence of $w$ in $B$, say $w'$. In the matrix of the eBWT, all circular rotations of $w$ and $w'$, since they are identical, will be grouped together; and their last letters will be consecutive, and equal, in the eBWT. Therefore, removing $w$ from $A$ will eliminate consecutives duplicates of letters, and the number of runs will remain unchanged.

(2) Now, suppose that $m=1$. In $\texttt{ebwt}(A)$, there are $\size{w}$ letters corresponding to $w$. Removing the circular rotations of $w$ from the eBWT matrix of $A$ leads to the eBWT matrix of $B$, and, importantly, does not modify the relative order of the circular rotations of the remaining strings. It remains to quantify the impact on the number of runs when a single row is removed from the eBWT matrix of $W$. In the worst case, all circular rotations of $w$ are sandwiched between circular rotations of other strings. For each of these sandwiches, the eBWT is locally modified from $\cdots abc\cdots $ to $\cdots ac\cdots $ when removing the letter $b$. The number of associated runs before removing $b$ is equal to $\mathds{1}_{a\neq b}+\mathds{1}_{b\neq c}$, whereas after removal it is equal to $\mathds{1}_{a\neq c}$. If $a=c$, then the number of runs in $A$ is either $2$ (if $b\neq a$) or $0$ (if $b=a$), and $0$ in $B$. If $a\neq c$, the number of runs in $A$ is either $2$ (if $a\neq b\neq c$) or $1$ (if $a=b\neq c$ or 
$a\neq b=c$) and $1$ in $B$. Eitherway, the  number of runs can only \emph{decrease}, therefore $\rho(A)\geq \rho(B)$, and by at most $2$. Since this occurs, in the worst case, for each letter of $w$, we indeed have $0\leq \rho(A)-\rho(B) \leq 2\cdot \size{w}$.
\end{proof}

From Proposition~\ref{prop:remove_one_occurrence}, we immediately conclude the two following results.

\begin{corollary}\label{corollary:remove_duplicates}
Let $A$ be a multiset of strings, and $B$ the associated set (without duplicates). Then $\rho(A)=\rho(B)$.
\end{corollary}
\begin{proof}
Apply repeatedly item (1) of Proposition~\ref{prop:remove_one_occurrence} until all duplicates are gone.
\end{proof}

\begin{corollary}\label{corollary:inclusion}
Let $A,B$ be two sets of strings with $B\subseteq A$; then $\rho(B)\leq \rho(A)$.  
\end{corollary}
\begin{proof}
Since $B$ can be obtained from $A$ by removing the only occurrence of each element of $A\setminus B$, we apply item (2) of Proposition~\ref{prop:remove_one_occurrence} to get $\rho(A)-\rho(B)\geq 0$.
\end{proof}

\subsection{Proof of Theorem~\ref{thm:best}}

We start by the following result.
\vspace{-\baselineskip}
\begin{lemma}\label{lemma:all_strings}
For any $p\geq 1$, $\texttt{ebwt}(\Sigma^p) = \overbrace{(a_1^p\cdots a_\sigma^p)\cdots (a_1^p\cdots a_\sigma^p)}^{\sigma^{p-1} \text{ times}}$. It follows that $\rho(\Sigma^p)=\sigma^p$.
\end{lemma}
\begin{proof}
Since we are calculating the eBWT of all the strings in $\Sigma^p$, the matrix of the eBWT, containing all of their circular rotations, is made up of $p$ consecutive copies of each of the $\sigma^p$ strings in $\Sigma^p$.

Fix a string $w$ of $\Sigma^{p-1}$. In the eBWT matrix, we find $p$ times the string $wa_1$, followed by $p$ times the string $wa_2$, and so on up to $p$ times the string $wa_\sigma$. Therefore the string $w$ contributes, in the last column of the matrix, to the sequence of letters $a_1^p\cdots a_\sigma^p$. Since there are $\sigma^{p-1}$ strings in $\Sigma^{p-1}$, the claim holds. Computing the number of runs is straightforward.
\end{proof}

The next result then follows naturally. 
\begin{corollary}\label{cor:multiset_equal_size}
Let $A=\lbrace\!\lbrace w_1,w_2,\dots\rbrace\!\rbrace$ be a multiset of strings with $\forall i, |w_i|=p$; then $\rho(A) \leq \sigma^{p}$.
\end{corollary}
\begin{proof}
We start by removing duplicates from $A$, obtaining the set $B=\lbrace w_1,w_2,\dots\rbrace \subseteq \Sigma^p$. We have $\rho(A)=\rho(B)$  using Corollary~\ref{corollary:remove_duplicates} and $\rho(B)\leq \sigma^p$ using Corollary~\ref{corollary:inclusion} and Lemma~\ref{lemma:all_strings}.
\end{proof}

We now introduce the principal contribution of this section.

\begin{proposition}\label{prop:linear_euclidean_partition}
Let $p\geq1$ and $w=a_1\cdots a_n\in \Sigma^n$ be a string, with $n= pq+r$ and $0\leq r <p$. Let $A= \lbrace\!\lbrace w_1,\dots,w_q\rbrace\!\rbrace \in \mathcal{D}_{p-1}(w)$ where
$$\begin{cases}
w_i &= a_{(i-1)p+1}\cdots a_{ip}  \text{ for } 1\leq i \leq q-1\\
w_{q} &=a_{(q-1)p} \cdots a_{pq}\cdots a_{pq+r}
\end{cases},$$
then $\rho(A) \leq \sigma^{p} + 2(p+r)$.
\end{proposition}
\begin{proof}
First, let $B=\lbrace\!\lbrace w_1,\dots,w_{q-1}\rbrace\!\rbrace$. Using Corollary~\ref{cor:multiset_equal_size}, we have $\rho(B)\leq \sigma^p$ since $|w_i|=p$ for $1\leq i\leq q$. Since $B$ is obtained from $A$ by removing the only occurrence in $A$ of $w_{q}$, we apply item (2) of Proposition~\ref{prop:remove_one_occurrence} to get
$\rho(A)-\rho(B)\leq 2\cdot \size{w_{q}} = 2(p+r)$.\end{proof}

With regard to the proof of Theorem~\ref{thm:best}, we derive that, with $p=k+1$, for any string $w$, since $A\in \mathcal{D}_k(w)$ and $r\leq k$, $\displaystyle\min_{W\in\mathcal{D}_k(w)} \rho(W)\leq \sigma^{k+1}+4k+2$.

\section{On the antecedents of $(ba)^n$ with the eBWT}\label{sec:antecedents_ba}

Let $n\geq 1$ be some integer. We consider in this section the multiset of strings $W(n)=\lbrace\!\lbrace w_1,w_2,\dots \rbrace\!\rbrace$ on the binary alphabet $\Sigma=\lbrace a,b\rbrace$ so that $\texttt{ebwt}(W(n))=(ba)^n$. Note that $W(n)$ is well defined and exists for any value of $n$, and that $\sum_{w\in W(n)} |w|=2n$. Moreover, $\rho(W(n)) = \texttt{runs}((ba)^n)=  2n-1$.

More precisely, for $k\geq 1$ fixed, we are interested in characterizing the values of $n$ for which the strings composing $W(n)$ are all of length at least $k+1$, i.e. so that $\min_{w\in W(n)} |w| > k$. In this section, we prove the following result.

\begin{theorem}\label{thm:infinite_n}
For any $k\geq 1$, there are infinitely many values of $n\geq 1$ for which $\displaystyle\min_{w\in W(n)} |w|> k.$
\end{theorem}

Therefore, concatenating the strings of $W(n)$ leads to a string $w$ of size $2n$, who admit a $k$-restricted decomposition -- $W(n)$ -- so that $\rho(W(n))=\size{w}-1$, which is maximal and allows to conclude the proof of Theorem~\ref{thm:ratio}.

\emph{First attempt.} A straightforward way to prove Theorem~\ref{thm:infinite_n} would be to exhibit an infinite number of values of $n$ for which $\size{W(n)}=1$, since we would have $\min_{w\in W(n)} \size{w} = 2n >k$ for $n$ large enough. Unfortunately, the existence of such an infinite sequence is linked to a conjecture by Artin from 1927, which remains unsolved to this day \cite{moree2012artin}. More details can be found in Appendix~\ref{app:artin_conjecture}.

The rest of this section makes extensive use of the process to invert the eBWT detailed in Section~\ref{sec:definition}, in order to determine the multiset of strings $W(n)$.

\emph{Structure of $L$ and $F$.} Remember that $L$ and $F$ are, respectively, the last and the first column in the eBWT matrix. In our context, $L = (ba)^n$ and $F=a^nb^n$. We number each of the letters $a$ and $b$ according to the order in which they appear in $L$ and $F$. Note the following :
\begin{itemize}
    \item $a_i$ is in position $i$ in $F$ and $2i$ in $L$;
    \item $b_i$ is in position $n+i$ in $F$ and $2i-1$ in $L$.
\end{itemize}

\emph{Proof for $k=1$.} If $k=1$, we want to prevent a letter in $L$ from being its own antecedent in $F$. This would imply, for some $1 \leq i \leq n$, that $i=2i$ if such a letter were $a_i$; or that $n+i=2i-1$ if it were $b_i$. Both case are absurd so for $k=1$, any value of $n$ is acceptable.

\emph{Proof for $k=2$.} For some $1\leq i,j\leq n$, a cycle of length $2$ during the inversion of the eBWT would be of the form $a_i\to b_j\to a_i$, as seen below left, and would verify the system provided below right.

\begin{center}
\begin{minipage}[c]{0.3\textwidth}
\centering
        \begin{tikzpicture}[yscale=0.5]

\node at (0,6) {$F$};
\node at (2,6) {$L$};

\draw (-0.25,5.5)--(2.25,5.5);

\node (fa) at (0,5) {$a_i$};
\node (fb) at (0,4) {$b_j$};

\node (lb) at (2,5) {$b_j$};
\node (la) at (2,4) {$a_i$};

\tikzstyle{arc}=[->, thick,lred]
\tikzstyle{retour}=[->, thick,lred,dotted]

\draw[arc] (lb)--(fa);
\draw[retour] (fa)--(la);
\draw[arc] (la)--(fb);
\draw[retour] (fb)--(lb);

\end{tikzpicture}
\end{minipage}~
\begin{minipage}[c]{0.3\textwidth}
\vspace{-0.5\baselineskip}
     $$\begin{cases}
i &= 2j-1\\
n+j &= 2i
\end{cases}$$
\end{minipage}
\end{center}
The system is solved by $i = \frac{2n+1}{3}$ and $ j=\frac{n+2}{3}$ hence such a cycle is possible only if $n \equiv 1 \mod 3$. Therefore, to forbid cycles of length $2$, it suffices to have $n\not\equiv 1\mod 3$, for which an infinite number of values are indeed possible.

\emph{Subsequent values.} Fix some $k\geq 3$ and let $1\leq i_1,\dots,i_k\leq n$. A cycle of length \emph{exactly} $k$ is necessarily of the form $a_{i_1} \to x_{i_2} \to \cdots \to x_{i_{k-1}} \to b_{i_k} \to a_{i_1}$ where $x_{i_j}\in \lbrace a_{i_j},b_{i_j}\rbrace$. Moreover, if we partition the indices $i_1,\dots,i_k$ according to whether the associated letter is an $a$ or a $b$, then each of the two elements of the partition must not contain duplicates for the cycle to be of length exactly $k$. With the notation $t_j = \begin{cases}
    1 &\text{if } x_{i_j} = b_{i_j},\\
    0 &\text{otherwise;}
\end{cases}$ and with the convention $t_1=0$ and $t_k=1$, this non-duplicates condition translates into
\begin{equation}\label{eq:non_duplicates}
 |\lbrace i_j : t_j =1\rbrace| = \sum_{j=1}^k t_j \quad \text{ and }\quad |\lbrace i_j : t_j = 0\rbrace| =k - \sum_{j=1}^k t_j.
\end{equation}

For the aforementioned cycle to exists, the indices $i_1,\dots,i_k$ would also need to verify the following system:
$$\begin{cases}
nt_j + i_j&= 2i_{j+1}-t_{j+1} \quad \forall 1\leq j\leq k-1\\
nt_k+i_k &= 2i_1 - t_1
\end{cases}$$
which is best represented in matrix form as
\begin{equation}\label{eq:system_matrix}
n\begin{pmatrix} t_1 \\
\vdots\\
t_k
\end{pmatrix}+\begin{pmatrix}
 i_1\\
 \vdots\\
 i_k
\end{pmatrix}=2\begin{pNiceMatrix}
0 & 1 & 0 & \Cdots&  0\\
&\Ddots & \Ddots & \Ddots& \Vdots\\
\Vdots & & &&0\\
0& \Cdots &&0 &1\\
1 & 0 & \Cdots&&0
\end{pNiceMatrix}
\begin{pmatrix}
i_1\\
\vdots\\
i_k
\end{pmatrix}-\begin{pNiceMatrix}
0 & 1 & 0 & \Cdots&  0\\
&\Ddots & \Ddots & \Ddots& \Vdots\\
\Vdots & & &&0\\
0& \Cdots &&0 &1\\
1 & 0 & \Cdots&&0
\end{pNiceMatrix}\begin{pmatrix}
t_1\\
\vdots\\
t_k
\end{pmatrix}.
\end{equation}

Denoting by $\mathbf{t}$ the vector $(t_1,\dots,t_k)$, $\mathbf{i}$ the vector $(i_1,\dots,i_k)$ and $S$ the binary matrix, \eqref{eq:system_matrix} is equivalent to
$$n\mathbf{t}+\mathbf{i} = 2S \cdot\mathbf{i}-S\cdot\mathbf{t} \iff \mathbf{i} = (2S-I)^{-1} \cdot(nI+S)\cdot \mathbf{t},$$
provided the matrix $2S-I$ is invertible. We have
$$2S-I = \begin{pNiceMatrix}
-1 & 2 & 0 & \Cdots&  0\\
0&\Ddots & \Ddots & \Ddots& \Vdots\\
\Vdots & \Ddots& &&0\\
0& \Cdots &0 &&2\\
2 & 0 & \Cdots&0 &-1
\end{pNiceMatrix}.$$

We recognize a circulant matrix \cite{kra2012circulant} of the form
$$C(c_0,\dots,c_{k-1})=\begin{pmatrix}
c_0 & c_1 & c_2 & \dots & c_{k-1}\\
c_{k-1} & c_0 & c_1 & \dots& c_{k-2}\\
c_{k-2} & c_{k-1} & c_0 &\dots &c_{k-3}\\
\vdots & \vdots &\vdots&\ddots&\vdots\\
c_1 & c_2 & c_3 &\dots & c_0
\end{pmatrix}$$
where $c_0=-1$, $c_1=2$ and $c_2 = \dots = c_{k-1}= 0$. Note that the general term of any circulant matrix $C(c_0,\dots, c_{k-1})$ is given by $c_{(j-i\mod k)}$.

\begin{lemma}\label{lemma:matrix_invertible}
$2S-I$ is invertible and
$(2S-I)^{-1} = \displaystyle\frac{1}{2^k -1} C(1,2,\dots,2^{k-1}).$
\end{lemma}
\begin{proof}
The proof is deferred to Appendix~\ref{app:proof_matrix_invertible}.
\end{proof}

The solution to \eqref{eq:system_matrix} is therefore given by 
\begin{alignat*}{3}
   && \mathbf{i} &= (2S-I)^{-1} \cdot (nI+S)\cdot \mathbf{t}\\
   &\iff& \mathbf{i} &= \frac{1}{2^k-1} \Biggl(n \cdot C(1,2,\dots,2^{k-1}) \cdot  \mathbf{t} +  C(1,2,\dots,2^{k-1}) \cdot S \cdot \mathbf{t}\Biggl)\\
      &\iff& \mathbf{i} &=  \frac{1}{2^k-1} \Biggl(n \cdot  C(1,2,\dots,2^{k-1}) \cdot \mathbf{t} +  C(2^{k-1},1,\dots,2^{k-2}) \cdot \mathbf{t}\Biggl)
\end{alignat*}
noticing that $C(c_1,\dots,c_k)\cdot S = C(c_k,c_1,\dots,c_{k-1})$. Going back to the variables $i_j$, with $1\leq j\leq k$, we get
\begin{equation}\label{eq:solution_systeme}
i_j = \frac{\left(\displaystyle\sum_{l=1}^k 2^{(l-j\mod k)}\cdot t_l\right) n + \left(\displaystyle\sum_{l=1}^k 2^{(l-j-1\mod k)} \cdot t_l\right)}{2^k-1}, 
\end{equation}
where $(l-j\mod k)$ and $(l-j-1\mod k)$ are to be chosen in the range $[\![0,k-1]\!]$ in case of negative values. We rewrite \eqref{eq:solution_systeme} as
$$i_j = \frac{\alpha_j \cdot n + \beta_j}{2^k-1}.$$
Recall that $t_1=0$ and $t_k=1$. Therefore, $\mathbf{t}=(t_1,\cdots,t_k)$ can neither be $(0,\dots,0)$ nor $(1,\dots,1)$. Hence, $0<\alpha_j,\beta_j < 2^k-1$.

Remember that, for a cycle of length exactly $k$ to exist, we must have (i) $i_j\in\mathbb{N}$, (ii) $1\leq i_j\leq n$  and (iii) equation \eqref{eq:non_duplicates} must hold. Each of these conditions is a necessary condition. It is therefore sufficient to break just one of them to guarantee that no cycle of length exactly $k$ can exist. (i) is equivalent to $\alpha_j\cdot n + \beta_j \equiv 0 \mod 2^k-1$. Since $\beta_j\not\equiv 0\mod 2^k-1$, it suffices to choose $n\equiv 0 \mod 2^k-1$ to ensure that $i_j\not\in\mathbb{N}$.

Therefore, in the context of Theorem~\ref{thm:infinite_n}, since we want to forbid the presence of any cycle of length $\leq k$, it suffices to choose 
$$n\equiv 0 \mod \prod_{k'=2}^k (2^{k'}-1),$$
for which there is indeed an infinite number of values, as claimed.

\section*{Acknowledgements}

F.I. is funded by a grant from the French ANR: Full-RNA ANR-22-CE45-0007.

\bibliography{main}

\appendix

\section{When $W(n)$ is reduced to a single string}\label{app:artin_conjecture}

Theorem~\ref{thm:infinite_n} would be straightforward if there were an infinite number of values of $n$ such that $|W(n)|=1$, since then we would have $\min_{w\in W(n)} |w| = 2n > k$ for $n$ large enough. Whenever $|W(n)|=1$, we have $\texttt{ebwt}(W(n))=\texttt{bwt}(W(n))$, and therefore the associated values of $n$ corresponds to the ones where the string $(ba)^n$ admits an antecedent with the BWT. A proper characterization of these values of $n$ was given in \cite[Proposition~4.3]{mantaci2017measuring}, as reproduced below.

\begin{proposition}[Mantaci et al., 2017]\label{prop:mantaci}
$(ba)^n$ admits an antecedent with the BWT if and only if $n+1$ is an odd prime number and $2$ generates the multiplicative group $\mathbb{Z}_{n+1}^\ast$.
\end{proposition}

The first values of $n$ satisfying the conditions of Proposition~\ref{prop:mantaci} are
$$2, 4, 10, 12, 18, 28, 36, 52, 58,\dots $$
and correspond to the sequence of integers $n$ such that $n + 1$ belongs to the sequence \href{https://oeis.org/A001122}{A001122} of the OEIS\footnote{OEIS Foundation Inc. (2025), The On-Line Encyclopedia of Integer Sequences, Published electronically at \url{https://oeis.org.}}. Unfortunately, it is unknown whether this sequence is infinite or not. Emil Artin conjectured in 1927 that this sequence is infinite, but no proof has yet been established \cite{moree2012artin}. Therefore, we cannot conclude about Theorem~\ref{thm:infinite_n}; however, we thought useful to mention this direction, since a solution to Artin's conjecture would make the result immediate.

\section{Proof of Lemma~\ref{lemma:matrix_invertible}}\label{app:proof_matrix_invertible}
Let us denote $P=2S-I$, $Q =  \frac{1}{2^k -1} C(1,2,\dots, 2^{k-1})$ and $R=PQ$. We have $P = C(c_0,\dots,c_{k-1})$ with $c_0=-1$, $c_1=2$ and $c_2=\dots=c_{k-1} = 0$. To simplify notations, let $d_j = 2^{j}$ so that $Q = \frac{1}{2^k -1} C(d_0,\dots,d_{k-1})$. Finally, remember that the general term $C_{i,j}$ of a circulant matrix $C(c_0,\dots, c_{k-1})$ is given by $c_{(j-i\mod k)}$. 

We identify $R$ with the identity matrix. We have $$R_{i,j} = \sum_{l=1}^k P_{i,l} \cdot Q_{l,j} =  \sum_{l=1}^k\frac{c_{(l-i\mod k)}\cdot  d_{(j-l\mod k)}}{2^k-1}.$$
Since $c_0 = -1$, $c_1=2$ and $c_2 = \dots =c_{k-1}=0$, we have $l-i \equiv 0 \mod k \iff l\equiv i \mod k \iff l=i$ and $l-i \equiv 1 \mod k\iff l\equiv i+1 \mod k$, leading to
$$R_{i,j} = \frac{2d_{(j-i-1\mod k)}-d_{(j-i\mod k)}}{2^k-1}.$$

Remember that $d_j = 2^{j}$. This gives us, for $i=j$,
$$R_{i,i} = \frac{2d_{(i-i-1\mod k)}-d_{(i-i\mod k)}}{2^k-1} = \frac{2d_{k-1}-d_{0}}{2^k-1}=1$$
and, for $i\neq j$, denoting $p=j-i\mod k$ and noticing that $1\leq p\leq k-1$, 
$$R_{i,j} = \frac{2d_{(p-1 \mod k)}-d_{(p\mod k)}}{2^k-1} = \frac{2d_{p-1}-d_{p}}{2^k-1}=\frac{2\cdot 2^{p-1}-2^{p}}{2^k-1}=0.$$
Therefore, $R=I$ and $Q=P^{-1}$.

\end{document}